\documentclass[a4paper,fleqn]{cas-dc}
\usepackage{amsmath,amsfonts}
\usepackage{amssymb,tikz}
\usepackage{algorithmic}
\usepackage{algorithm}
\usepackage{array}
\usepackage{epsfig}
\usepackage{amsthm}
\usepackage{hyperref}
\newtheorem{theorem}{Theorem}
\newtheorem{definition}{Definition}
\newtheorem{observation}{Observation}
\newtheorem{lemma}{Lemma}
\usepackage[caption=false,font=normalsize,labelfont=sf,textfont=sf]{subfig}
\usepackage{textcomp}
\usepackage{stfloats}
\usepackage{url}
\usepackage{verbatim}
\usepackage{graphicx}
\usepackage[numbers]{natbib}

\def\tsc#1{\csdef{#1}{\textsc{\lowercase{#1}}\xspace}}
\tsc{WGM}
\tsc{QE}
\tsc{EP}
\tsc{PMS}
\tsc{BEC}
\tsc{DE}

\begin{document}
\let\WriteBookmarks\relax
\def\floatpagepagefraction{1}
\def\textpagefraction{.001}
\shorttitle{Maximum Exposure Problem}
\shortauthors{Remi Raman et. al.}

\title [mode = title]{On the Parameterized Complexity of the Maximum Exposure Problem}                      



\author[1]{Remi Raman}
\ead{remi_p170027cs@nitc.ac.in}



\author[1]{Shahin John J S}
\ead{ivinjohn98@gmail.com}

\author[1]{ R Subashini}
\ead{suba@nitc.ac.in}

\author[1]
{Subhasree Methirumangalath}
\ead{subha@nitc.ac.in}



\begin{abstract}
We investigate the parameterized complexity of Maximum Exposure Problem (MEP). Given a range space ($R, P$) where $R$ is the set of ranges containing a set  $P$ of points, and an integer $k$, MEP asks for $k$ ranges which on removal results in the maximum number of exposed points. A point $p$ is said to be exposed when $p$ is not contained in any of the ranges in $R$. The problem is known to be NP-hard. In this letter, we give fixed-parameter tractable results of MEP with respect to different parameterizations. 
\end{abstract}


\maketitle

\section{Introduction}

Consider $n$ sensors deployed for tracking assets in an area. The topology of sensing zones and assets can be described using a range space ($R, P$) where the sensing zones are represented by ranges in $R$ and assets are represented by the points in $P$. We are interested in finding the maximum number 
of exposed assets when $k$ number of sensors are compromised. 
From the computational geometry point of view, this problem
is the Maximum Exposure Problem (MEP) \cite{KUMAR2022101861}.
Given a range space ($R, P$), MEP asks for a set of $k$ ranges,
which on removal results in a maximum number of exposed
points. A point p is said to be exposed, if it is not contained
in any of the ranges. In Figure \ref{fig:my_input}, the ranges are axis-aligned  rectangles  and  if $k$ = 2,  one of the  solutions  is $R_4$ and $R_5$ that exposes 10 points.\\
MEP is shown to be NP-hard and it is also hard to approximate even when ranges in $R$ are translates of two fixed rectangles \cite{KUMAR2022101861}. The authors also proposed a polynomial-time approximation scheme (PTAS) if $R$ only consists of translates of a single rectangle. For polygons with constant number of sides and for arbitrary disks, an $O(k)$ bicriteria approximation algorithm is also presented. 
They also showed that MEP with convex polygin ranges is equivalent to the densest $k$-subhypergraph problem on a dual hypergraph $H = (X, E)$, where the ranges in $R$ corresponds to vertices in $X$
and the set of points $P$ corresponds to edges in $E$. 
Given a hypergraph $H = (X, E)$, the densest $k$-subhypergraph problem finds a set of $k$ vertices with the maximum number of induced hyperedges. \\
Parameterized complexity\cite{cygan2015parameterized,downey2013fundamentals,niedermeier2006invitation} offers a methodology
for solving NP-hard problems by expressing their running time in terms of one or more
parameters, in addition to the input size. Fixed parameter tractable algorithms have their running times of the form $O(f(k).n^c)$ where $n$ is the size of the input instance, $k$ is a non-negative integer parameter, $f$ is a computable function depending only on $k$, and $c$ is a constant. The problems for which we can find such algorithms are referred to as \textit{fixed parameter tractable} (FPT).
\noindent
In addition, to deal with problems that  that do not
admit a fixed parameter tractable algorithm,
Downey and Fellows defined a fixed parameter reduction and a hierarchy of classes $W[1] \subseteq W[2] \subseteq \ldots$ that includes fixed parameter intractable problems. 

We study parameterizations of MEP with respect to two different parameters and a combination of two parameters. We use the greediness of parameterization approach proposed by Bonnet et. al. \cite{bonnet2015multi}
which is based on branching algorithms. This approach has recently gained much attraction in fixed cardinality parameterized problems \cite{lokshtanov2019balanced, saurabh2019parameterized, shachnai2017parameterized}.
We also use the result by Lenstra and Kannan \cite{kannan1987minkowski, lenstra1983integer}, which shows that ILP optimization is fixed parameter tractable when parameterized by the number of variables.\\
\noindent
The rest of the paper is organized as follows. In section 2, the problem definitions and a brief overview of the results in this letter are provided. Section 3 gives the two FPT results of MEP using greediness of parameterization approach and integer linear programming approach.

\section{Our Results}
\noindent
We consider the following parameterized versions of MEP.

\begin{definition}[\textbf{$k$-MEP}] Given a range space ($R, P$) where $R$ is the set of ranges and $P$ is the set of points, and an integer parameter $k$ $\geq$ 1, $k$-MEP asks for $k$ ranges in $R$ which on removal results in maximum number of exposed points.
\end{definition}

\begin{definition}[\textbf{($l,k$)-MEP}] Given a range space ($R, P$) where $R$ is the set of ranges and $P$ is the set of points, with each range overlapping on at most $l$ other ranges and an integer parameter $k$ $\geq$ 1, ($l,k$)-MEP asks for $k$ ranges in $R$ which on removal results in maximum number of exposed points.
\end{definition}
In the following definition, we refer to the polygonal regions formed by the range space as \textit{cells}, and the set of cells that contains points as $D$. In Figure \ref{fig:my_input}, the ranges $R_1$ and $R_2$ alone forms three cells:  $R_1$ $\setminus$ $R_2$, $R_1$ $\cup$ $R_2$, and $R_2$ $\setminus$ $R_1$.

\begin{definition}[\textbf{$d$-MEP}] Given a range space ($R, P$) where $R$ is the set of ranges and $P$ is the set of points, and the number of cells in the range space that contain points $|D|$ = $d$, $d$-MEP asks for $k$ number of ranges in $R$ which on removal results in maximum number of exposed points.
\end{definition}
\noindent
Our first result presented in \hyperref[lemma1]{Lemma 1} shows that $k$-MEP is W[1]-hard, which directly follows from the following observations:
\begin{observation} \cite{KUMAR2022101861}
MEP for convex polygons is equivalent to the densest-$k$- subhypergraph  problem for hypergraphs.
\end{observation}
\begin{observation}\cite{cai2008parameterized}
Densest k-subgraph problem is not FPT, i.e., it is W[1]-hard, with respect to $k$ even for regular graphs.
\end{observation}

\begin{lemma}
\label{lemma1}
MEP is W[1]-hard when parameterized with the solution size $k$.
\end{lemma}
 
 For any polygonal ranges with $O(1)$ complexity  (with respect to number of sides of the polygon) or circular disks, we show that the problem is FPT with respect to $l$ and $k$ where $l$ is the maximum number of overlaps in any range (using greediness of parameterization).
Also the problem is FPT with respect to $d$, the number of cells in the range space which contain points (using \textit{integer linear programming}).

\section{Fixed Parameter Tractable Algorithms}

We restrict ourselves to polygon ranges of O(1) complexity and circular disks, since the maximum number of cells that can be formed by an arrangement of $n$ such polygons/disks is $O(n^2)$ (Any arrangement of $n$ line segments/ circles have $O(n^2)$ cells\cite{de1997computational,asadzadeh2011optimal}).  
As an initial preprocessing step, we delete all the points which are contained in more than $k$ ranges. These points are not going to be exposed by removing any of the subset of $k$ ranges.
In the algorithms proposed, we consider
only the set $D$ of cells in the range space with points, as the other cells does not affect the optimal solution. \\  
We define the subset of ranges in $R$ that forms a cell in $D$ as a $cluster$. 
After the preprocessing step, maximum size of a $cluster$ is $k$. Based on the cardinality of the subset of the ranges, we classify the clusters into different cluster types: $1$-cluster, $2$-cluster, $\ldots$, $k$-cluster. In Figure \ref{fig:my_input}, $R_1$, $R_2$ and $R_3$ forms a 3-cluster, and $R_4$ and $R_5$ forms a 2-cluster. Any range in the range space that forms a cell that contains points and is not an overlap region is a $1$-cluster. 
Also, we define $OL(R_i)$ as the set of ranges that overlaps with $R_i$. In Figure \ref{fig:my_input}, $OL(R_4)$ is $\{R_3, R_5\}$.
 
In the following two sections, we explain the fixed parameter tractable results of MEP in O(1) complexity polygon ranges and circular disks with respect to two different parameterizations.

\subsection{An FPT algorithm for ($l,k$)-MEP}

In this section, we use the greediness of parameterization approach proposed by Bonnet et al.\cite{bonnet2015multi} which is based on branching
algorithms.
At each level of the branching tree, a partial solution is extended. 
At the leaves, the optimum solution among all the solution is returned.
The crux of the approach is to
branch on a greedy extension of the partial solution and also on the \textit{neighbourhood} of the greedy extension.
We present a branching algorithm $RecMEP(T, k)$ which maintains a set $T$ of ranges which is initially empty. At each level of branching, $T$ is added with at most $k$ ranges until no more ranges can be added i.e. until $k=$ 0. 
 The basic idea behind our algorithm is the following. At each level of branching tree, we branch on different $i$-clusters where 1 $\leq$ $i$ $\leq$ $k$, which on removal maximizes the number of exposed points. If there are multiple clusters of type $i$ exposing the  maximum number of points at that level, one among them can be chosen randomly.  We also branch on the \textit{neighborhood} of these clusters, which are those ranges overlapping with the ranges in the $i$-cluster considered. A greedy criterion may not be always optimal.  However, if at each step either the greedily chosen $i$-cluster, or some of its overlapping ranges is part of an optimal solution,  then the branching tree has at least one leaf which is an optimal solution. 
 Overall description of our algorithm $RecMEP (T, k)$ is given in \hyperref[algo1]{Algorithm 1}. Initially $T$ is set to $\emptyset$. Also, for a given $T = \{R_1, R_2, \ldots, R_i\}$, we calculate $OL(T)$ as  $OL(T) = OL(R_1) \cup OL(R_2) \cup \ldots \cup OL(R_i).$ 
 \vspace{.2cm}
 \hrule \vspace{.1cm} \label{sec:algo}
 \noindent
 \textbf{Algorithm 1}: $RecMEP (T, k)$ \label{algo1}
 \hrule \vspace{.1cm}
 \begin{itemize}
     \item if $k$ $>$ 0, then
    \begin{itemize}
        \item for each $i$ varying from 1 to $k$ pick a greedy cluster $C_i$ $\in$ $R$ $\setminus$ $\{T \cup OL(T)\} $ that exposes the maximum number of points and call $RecMEP$ ($T \cup{C_i}$, $k-|C_i|$)
        \item for each range $R_j$ $\in$ $OL(T)\setminus T$, call $RecMEP$ ($T \cup{R_j}$, $k-1$).
    \end{itemize}
    \item If $k$ = 0, store $T$ as a feasible solution.
    \item Return a $T$ with maximum number of exposed points as the optimal solution.
 \end{itemize}
 \hrule 
 \vspace{.5cm}
\begin{figure}
    \centering
    \includegraphics[scale=.75]{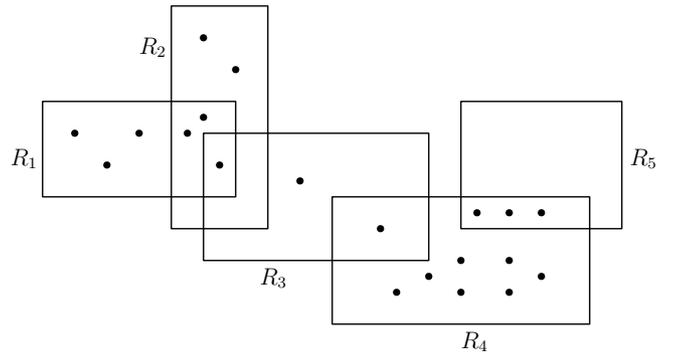}
    \caption{An Input MEP instance with 5 rectangular ranges. For $k$= 2, one of the solutions to MEP is $R_4$ and $R_5$ exposing 10 points.}
    \label{fig:my_input}
\end{figure}

\begin{figure}
     \centerline{\includegraphics[scale=.85]{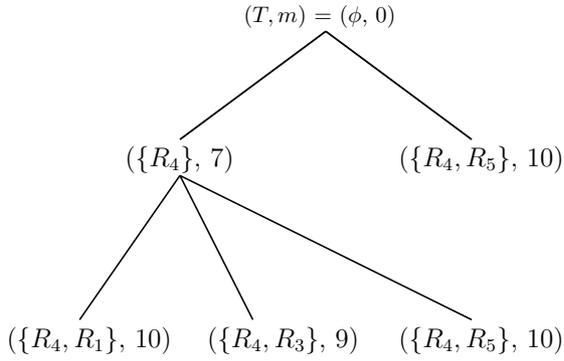}}
     \label{branch}
    \caption{The branching tree obtained for Figure \ref{fig:my_input} with $k$ = 2 and $l$ = 3. In each of the nodes, the set $T$ and the number of exposed points $m$ when $T$ is removed are shown. Here the solution is $\{R_4, R_1\}$ or $\{R_4, R_5\}$ that exposes 10 points}
    \label{fig:branch}
\end{figure}

Figure \ref{fig:branch} shows the branching tree of $RecMEP (T, k)$ for the input range space given in Figure \ref{fig:my_input}. We start with branching on greedy clusters of type one and two as the value of $k$ is two. We cannot further branch from the $2$-cluster $\{R_4, R_5\}$ as our budget $k$ is finished at the first level itself. From the $1$-cluster $\{R_4\}$, we branch on the next greedy $1$-cluster and add $R_1$ to $T$, and also on the two overlaps of the range $R_4$. 
If in the same input instance the value of $k$ = 3, we would have started with three greedy clusters of type 1, 2 and 3. In the next level, a greedy 1-cluster and a 2-cluster is also considered along with the overlapping ranges of $R_4$. Additionally, we would have branched on the chosen 2-cluster $\{ R_4, R_5 \}$ also. 

This branching process continues until the $k$ value becomes zero in all branches. From the leaves of the branching tree, the node with maximum number of exposed points is returned as the solution. 
\\  
Let us now establish the time complexity of the algorithm. 
The number of children of a node of the branching tree is at most $lk-l$. This happens when we branch on a ($k-1$)-cluster and each range in the cluster has $l$ neighbors. At each step, we add either a cluster or a neighbor, in either case depth of the branching tree is at most $k$, and the branching tree has size $O((lk-l)^k)$. On an internal node of the branching tree, algorithm only does polynomial computations and the running time is $O^*(lk-l)^k)$ \footnote{We use O* notation to hide polynomial factors with respect to $n$} or if $l$ $<$ $k$, $O^*(l^{2k})$-time (as cluster size is $min (l, k)$) , i.e., it is fixed parameter tractable with respect to $l$ and $k$.

\begin{theorem}
MEP is fixed parameter tractable with respect to $k$ and $l$, where $l$ is the maximum number of overlaps, a range creates with another one.
\end{theorem}  

\begin{proof}

Our proof of correctness is based on the hybridization method used in \cite{bonnet2015multi}. Let $T_0$ be a solution that exposes the maximum number of points. Recall that each node of the branching
tree has at most $k$ greedy clusters adding 1 to $k$ ranges to the solution and up to $lk-l$ neighbours each adding one range to the solution.\\ Let $B$ be the set of ranges in a maximal branch in the branching tree
from the root to a node $v$ such that all the ranges in $B$ are present in $T_0$. By the maximality of the branch, the range(s) in $v$ deviates from $T_0$. 
From $T_0 \setminus B$, find the largest $i$-cluster. None of the ranges in that $i$-cluster is a neighbour of $v$, as $B$ is the set of ranges in the maximal branch.  So, we can substitute the $i$-cluster in $T_0 \setminus B$ by the greedy $i$-cluster from $v$ as it will expose at least as many points as the other cluster. 
From $v$, we consider a maximal branch $B^{'}$ again in accordance with $T_0$, and we iterate the same hybridization method at most $k$  times until we reach a leaf. \\

\end{proof}
  
\subsection{An FPT algorithm for $d$-MEP}

In this section, we solve $d$-MEP by formulating the problem as an integer linear program.

Given a range space ($R, P$) with $|D|$ = $d$ being the number of cells with points. For any cell $c_i$ in $D$, we use $cluster(c_i)$ to denote the set of ranges that forms the cell $c_i$ and $points(c_i)$ for all the points that are contained in the $cluster(c_i)$. Let $S$ be a solution of MEP that on removal results in the maximum number of exposed points. In our algorithm we try to guess a set $C$ = $\{c_1, c_2, \ldots, c_i\}$ (for 1 $\leq$ $i$ $\leq$ $d$) which is a subset of $D$, such that the ranges in $\{cluster(c_1)$ $\cup$ $cluster(c_2)$ $\cup$ $\ldots$ $\cup$ $cluster(c_i)\}$ have a non-empty intersection to the set $S$. 
There are $2^{d}$ such subsets possible for this guess. \\ We then reduce the problem of finding a set $S$ of $k$ ranges which on removal maximizes the number of exposed points to \textit{integer linear programming} (ILP) optimizations with at most $d$ variables in each ILP optimization. 
The ILP optimization problem can be parameterized by the number of variables \cite{fellows2008graph} and, here we parameterize $d$-MEP by the number of cells $d$. Our idea in ILP is to maximize the number of points in the clusters of the subset guessed, subject to the condition that the number of ranges in the clusters is $k$.  Given a set of cells $C\subseteq \{c_1, c_2, \ldots, c_d\}$, we present the ILP formulation of $d$-MEP problem as follows:
\\ \\
\textbf{ILP Formulation}:\\
For each $c_i$, we associate two variables; $x_i$ that indicates $|S\cap cluster(c_i)|$ = $x_i$ and $y_i$ with $|S\cap points(c_i)|$ = $y_i$. The possibility of same ranges containing in different clusters is eliminated by using a Boolean variable $val(c_i)$. $val(c_i)$  is assigned to 0 when $cluster(c_i)$ is a subset of $cluster(c_j)$, where $c_i$, $c_j$ $\in$ $C$ and $i$ $\neq$ $j$, and 1, otherwise.

\begin{equation*}
\begin{array}{l@{}l}
\text{maximize}  \displaystyle\sum\limits_{c_i\in C}^{}val(c_{i}) y_{i} &\\
\text{subject to} \\ \displaystyle\sum x_{i} = k \\  \\ x_{i} \in \{1,2, \dots, |cluster(c_i)|\}     \text{ for all } c_i \in C
              \end{array}
\end{equation*}
\noindent
\textbf{Solving the ILP}:
    It is shown that the feasibility version of a $p$ variable ILP is fixed parameter tractable with running time doubly exponential in $p$ by Lenstra \cite{lenstra1983integer} and Kannan \cite{ kannan1987minkowski}. Frank and Tardos \cite{ frank1987application} improved the algorithm to run in polynomial space. In our case we use the optimization version of $p$-ILP defined by Fellows et.al. \cite{fellows2008graph}. They showed that the optimization version can be solved in $O^*(p^{2.5p}+o(p).L.log (MN))$ time and space polynomial in $L$. Here, $L$ is the number of bits in the input, $N$ is the
maximum of the absolute values any variable can take, and $M$ is an upper bound on the
absolute value of the minimum taken by the objective function. \\
In the formulation for $d$-MEP, we have at most $d$ variables, and $2^d$ such formulations. The value of objective function is bounded by $m$ which is the maximum number of points that can be exposed. The value of any variable in the integer linear programming is also bounded by max($m$, $n$). There are at most $2^d$ choices for $C$, and the ILP formulation for a guess can be solved in FPT time.
Therefore, the total runtime is $O^{*}(2^{d}d^{2.5d})$. Thus, the following theorem holds.

\begin{theorem}
MEP is Fixed parameter tractable with respect to the number of cells $d$ in the range space that contain points.
\end{theorem}

\section{Conclusion}

We show that the maximum exposure problem (MEP) is W[1]-hard with respect to the solution size $k$ and FPT with respect to the parameters ($l,k$) and $d$ where $l$ is the maximum number of overlaps in any range and $d$ is the number of cells in the range space
which contain points.
It will be interesting to investigate further variants of the problem to see which all restricted range spaces gives polynomial time solutions or FPT solutions. It would also be interesting to see if other multiparameterizations are possible for the general case.

\bibliography{ver1.bib}
\bibliographystyle{cas-model2-names}
\vfill



\end{document}